\documentclass[8pt]{article}
\usepackage{amsmath}
\usepackage{amssymb,amsfonts}
\usepackage{amsthm}
\usepackage{thmtools, thm-restate}
\usepackage{mathrsfs}
\usepackage{mathtools}
\usepackage{comment}
\usepackage{bbm}
\usepackage{authblk}

\declaretheorem[{style=definition,numberwithin=section}]{definition}
\declaretheorem[{style=definition,sibling=definition}]{theorem}
\declaretheorem[{style=definition,sibling=definition}]{lemma}

\declaretheorem[{style=definition,sibling=definition}]{corollary}

\title{On the Minimum Decoding Delay of Balanced Complex Orthogonal Design}

\author[1]{Xiaodong Liu\thanks{Email:xdliufdu@gmail.com}}
\author[2]{Yuan Li\thanks{Email:yuanli@cs.uchicago.edu}}
\author[1]{Haibin Kan\thanks{Email:hbkan@fudan.edu.cn}}
\affil[1]{Department of Computer Science, Fudan University}
\affil[2]{Department of Computer Science, University of Chicago}

\date{}

\begin{document}
\maketitle

\begin{abstract}
Complex orthogonal design (COD) with parameter $[p, n, k]$ is a combinatorial design used in space-time block codes (STBCs). For STBC, $n$ is the number of antennas, $k/p$ is the rate, and $p$ is the decoding delay.
A class of rate $1/2$ COD called balanced complex orthogonal design (BCOD) has been proposed by Adams \emph{et al}., and
they constructed BCODs with rate $k/p = 1/2$ and decoding delay $p = 2^m$ for $n=2m$. Furthermore,
they prove that the constructions have optimal decoding delay when $m$ is congruent to $1$, $2$, or
$3$ module $4$. They conjecture that for the case $m \equiv 0 \pmod 4$, $2^m$ is also
a lower bound of $p$. In this paper, we prove this conjecture.

\end{abstract}

\section{Introduction}
Since the pioneering work by Alamouti \cite{Ala98}, and the work by Tarokh \emph{et al}. \cite{TJC99}, \emph{complex orthogonal designs} (CODs) have become an effective technique for the design of \textit{space-time block codes} (STBCs). The importance of this class of codes is due to the fact that they achieve full diversity and have the fast maximum-likelihood decoding.

A COD $G[p, n, k]$ is a $p \times n$ matrix,
where each nonzero entry is either $\pm z_i$ or $\pm z^*_i$, $i = 1, 2, \ldots, k$ such that
$$
G^H G = I_n (|z_1|^2 + \ldots + |z_k|^2).
$$
For the application as STBCs, linear combination can be allowed, that is,
each entry is a complex linear combination of $z_1, \ldots, z_k$ and their conjugations, which is called \textit{generalized COD} sometimes.

Motivated by STBCs, we are interested in certain criterions of COD including \textit{rate} $k/p$, which is the ratio of number of variables transmitted and time units $p$; \textit{decoding delay}, which is the number of rows $p$; \textit{transceiver signal linearization}, which can be achieved if the variables in any row are either all conjugated or nonconjugated \cite{SBP04} etc.

It's impossible to optimize all the design considerations simultaneously for general $n$. For the rate, Liang \cite{Lia03} proves a tight
upper bound, which is slightly greater than $1/2$ depending on $n$. When the rate is maximized, Adams \emph{et al}. proves a lower
bound on the decoding delay, which grows factorially in $n$ \cite{AKM10, AKP07}. Furthermore, Kan and Li \cite{LK12} give a complete classification of  ``first type'' CODs, which are those without submatrices
$\left(
  \begin{array}{cc}
    \pm z_j & 0 \\
    0 & \pm z_j^* \\
  \end{array}
\right)$
and
$\left(
  \begin{array}{cc}
    \pm z_j^* & 0 \\
    0 & \pm z_j \\
  \end{array}
\right),$
which contains all the CODs with maximum rate.

The delay for maximum rate CODs grows quickly as the number of antennas increases. It might be possible to significantly
lower down the decoding delay at the cost of decreasing the rate a little bit. For this purpose, Adams \emph{et al}. construct a class of rate 1/2 CODs with decoding delay $p=2^m$ when $n=2m$, which are called \textit{balanced complex orthogonal designs} (BCODs). They also prove that $2^m$ is the lower bound of the decoding delay when $n \equiv 1, 2, 3 \pmod 4$; $2^{m-1}$ when $n \equiv 0 \pmod 4$. They conjecture $2^m$ is also a lower bound when $m \equiv 0 \pmod 4$.  In this paper, we prove the conjecture. Our proof is combinatorial. Although the presentation is self-contained, the concepts and proof techniques heavily depend on the techniques developed in those aforementioned papers.

We organize the paper as follows. In Section II, we introduce some definitions and notations. In Section III, we define and
study the properties of the \textit{standard form} of BCOD. In Section IV, we prove the tight lower bound for the delay of BCOD, depending on some properties of the standard form.

\section{Definitions and notations}
\begin{definition} \cite{TJC99} A \textit{complex orthogonal design} (COD) $G[p, n, k]$ is a $p \times n$ matrix,
where each nonzero entry is either $\pm z_i$ or $\pm z^*_i$, $i = 1, 2, \ldots, k$ such that
$$
G^H G = I_n (|z_1|^2 + \ldots + |z_k|^2),
$$
where $G^H$ denotes the \emph{Hermitian transpose} of $G$.
\end{definition}

\begin{definition} \cite{ADK11}
A COD $G[2k,n,k]$ with $n=2m$ columns is a \textit{balanced complex orthogonal design} (BCOD) if it satisfies the following conditions.
\footnote{In \cite{ADK11}, there is an extra condition that ``For each $j=1,2,\cdots,k$, $z_j$ and $z_j^*$ each appears $m$ times (up to sign)'', which can be implied from the definition of COD and 2, 3.}

\begin{enumerate}
\item Every row of $G$ has exactly $m$ zeros and $m$ nonzero entries;
\item $G$ is conjugation separated;
\item For each $j \in [k]$, the $M_j$ submatrix of the $B_j$ form submatrix is skew-symmetric, i.e., $M_j^T = -M_j$.
\end{enumerate}
\end{definition}

Given some COD $G[p, n, k]$, it's easy to see that row or column permutations, negating or
conjugating some variables, etc., will not change the orthogonality, which is formalized in the following definition.

\begin{definition} Following operations on COD are called \textit{equivalence operations}.
\begin{itemize}
\item Rearrange the order of the rows (``row permutation'').

\item Rearrange the order of the columns (``column permutation'').

\item Conjugate all instances of a certain variable (``instance conjugation'').

\item Negate all instances of a certain variable (``instance negation'').


\item Multiply any row by $-1$ (``row negation'').

\item Multiply any column by $-1$ (``column negation'').
\end{itemize}
\end{definition}

Many results about COD comes from the following observation. The idea is that although COD is difficult to
understand in global, it has a very simple \emph{local} characterization.

\begin{definition} \cite{Lia03}
Fix some variable $z_j$. The following matrix is called \textit{$B_j$ form}:
$$
B_j = \left(
        \begin{array}{cc}
          z_j {I}_{n_1} & {M}_j \\
          -{M}_j^H & z_j^*  {I}_{n_2} \\
        \end{array}
      \right),
$$
where $n_1 + n_2 = n$, and $M_j$ is a $n_1 \times n_2$ matrix.
\end{definition}

It's easy to prove that $G[p, n, k]$ is a COD if each $z[j]$ appears in each column exactly once, and for each $j$,
the first $n$ rows is a $B_j$ form after equivalence operations. Based on this observation, Liang proves that for $n = 2m$
or $2m-1$, the rate $k/p$ is upper bounded by $(m+1)/(2m)$, and the bound is tight \cite{Lia03}. After that, Adams \emph{et al}.
solve the \emph{minimal delay} problem: what is the minimal $p$ when $k/p$ reaches the maximal?

\begin{theorem} \cite{AKM10, AKP07}
Let $n = 2m$ or $2m-1$. For COD $[p, n, k]$, if $k/p = (m+1)/(2m)$, then $p \ge \binom{2m}{m+1}$ when $n\equiv 0,1,3 \pmod 4$; $p \ge 2\binom{2m}{m+1}$ when $n\equiv 2 \pmod 4$.
\end{theorem}

Their proof is based on a new concept called \textit{zero pattern}. Let $r$ be some row in COD $[p, n, k]$. The \textit{zero pattern}
of $r$ is a vector in $\mathbb{F}_2^n$, where the $i$th entry is $0$ if and only if $r(i)$ is $0$, where $r(i)$ denotes
the element of $r$ on the $i$th column. Roughly speaking, the lower bound of $p$ is proved by showing the existence of all zero
patterns with some given weight. In order to investigate BCOD, we propose the following definition analogous to zero pattern.

\begin{definition} Let $r$ be some row in BCOD $G[2k, 2m, k]$, and $\alpha \in \mathbb{F}_2^{2m}$ be the zero pattern of $r$. The \textit{left zero pattern} of $r$ is
$$
\alpha_L = (\alpha(1), \alpha(2), \ldots, \alpha(m)) \in \mathbb{F}_2^m,
$$
and the \textit{left weight} of $r$ is the weight of $\alpha_L$.
\end{definition}

Our idea of proving the lower bound is similar to Adams \emph{et al}.'s, that is, to show the existence of some left zero patterns.
For example, when $n = 2m$, $m$ odd, we will prove that all zero patterns exist, where the total is $2^m$.



\vspace{0.2cm}
Besides $B_i$ form, there is another local characterization of orthogonality. Let $G[p, n, k]$ be some COD.
Consider all $2 \times 2$ submatrix of $G$ such that the diagonal elements are $z[i]$. It's easy to see that there
are only 3 possible cases (up to negation and conjugation).
$$
A = \left(
      \begin{array}{cc}
        z_i & z_j \\
        -z_j^* & z_i^* \\
      \end{array}
    \right)
$$
$$
D = \left(
     \begin{array}{cc}
         z_i & 0 \\
         0 & z_i^* \\
     \end{array}
    \right)
$$
$$
T = \left(
     \begin{array}{cc}
         z_i & 0 \\
         0 & z_i \\
     \end{array}
    \right)
$$
Call submatrix in the form of $A$ an \textit{Alamouti $2 \times 2$}, $D$ \textit{Diagonal $2 \times 2$}, and $T$ \textit{Trivial $2 \times 2$}. From the definition of BCOD, for some fixed $z[j]$ in some row of $G[2k, 2m, k]$, it's contained in $m-1$ Alamouti $2 \times 2$,
one Diagonal $2 \times 2$, and $m-1$ Trivial $2 \times 2$.

We introduce the following concept \textit{complement row}, which is important for proving the lower bound on the delay of BCOD.

\begin{definition}
Let $G$ be a $[2k,2m,k]$ BCOD, and $r$ be a row of $G$. If row $r_c$ satisfies:
\begin{enumerate}
\item $r$ and $r_c$ have complementary zero patterns.
\item $r$ and $r_c$ have opposite conjugations.
\item $r_c$ contains the same variables as $r$.
\end{enumerate}
Then we call $r_c$ the \textit{complement} of $r$.
\end{definition}

From the definition of BCOD and $B_j$ form, it's easy to verify that every $r$ has a unique complement \cite{ADK11}.

\begin{definition} \cite{LK12}
COD $G[2k,2m,k]$ is called atomic if and only if there does not exist a COD that is a submatrix of $G$ consisting of some (not all) rows of $G$.
\end{definition}

For an atomic BCOD $G[2k,2m,k]$, given any $1\leqslant s, t \leqslant k$, there exist $j_1 = s, j_2, \cdots, j_{m-1}, j_m = t$ such that $B_{j_1}$ and $B_{j_2}$ share some common rows, $B_{j_2}$ and $B_{j_3}$ share some common rows, $\cdots$, $B_{j_{m-1}}$ and $B_{j_m}$ share some common rows. This condition is also sufficient for a BCOD to be atomic. Clearly, a BCOD which achieves the minimum decoding delay must be atomic. In the following sequel, we will assume that all BCODs are atomic.

\section{Standard form}

In this section, we define a \emph{standard form} for BCOD and prove some properties, which will be used
in the proof of the lower bound of delay in the next section.

\begin{definition} Let $G$ be a BCOD $[2k, 2m, k]$. We say $G$ is in \textit{standard form} if and only if
it's already in $B_i$ form for some $i \in [k]$.
\end{definition}

Recall that we say $G$ is \emph{in $B_j$ form} if $B_j$ is a submatrix of $G$ after equivalence operations
\emph{without} column permutations.

\begin{definition} A sequence of equivalence operations are called \textit{column-restricted equivalence operations}
if all the column permutations are transpositions of column $i$ and $m+i$, for some $i \in [m]$.
\end{definition}

\begin{theorem}
\label{thm:main}
If BCOD $G[2k, 2m, k]$ is already in standard form. Then for any $j \in [k]$, $G$ can be
transformed into $B_j$ form by column-restricted equivalence operations.
\end{theorem}
\begin{proof} Without loss of generality, assume $G$ is already in $B_1$ form, i.e., there exist $2m$ rows
of the form
$$
B_1 = \begin{pmatrix}
z_1 I_m & M_1 \\
-M_1^H & z_1^* I_m \\
\end{pmatrix},
$$
where $M_1$ is skew-symmetric with diagonal all zeros.

Recall that we always assume $G$ is atomic. It suffices to prove that $G$ can be transformed into $B_j$ form for the adjacent $j$, that is, for those $z[j]$
in $M_1$. By the definition of BCOD, $G$ is conjugation separated, that is, all variables in $M_1$
are of the form $\pm z_j$. Take any variable $\pm z_j \in M_1$, and assume that $M_1(s, t) = \pm z_j$, where $s, t \in [m]$.
Since $M_1$ is skew-symmetric, that is, $M_1 = -M_1^T$, we have $M_1(t, s) = \mp z_j$.

On row $s$, there are $m$ zeros. Among the $m$ zeros, for $z_j$, there are $m-1$ Trivial $2 \times 2$
$
\begin{pmatrix}
0 &  \pm z_j \\
\pm z_j & 0 \\
\end{pmatrix},
$
 and $1$ Diagonal $2 \times 2$
 $
\begin{pmatrix}
0 &  \pm z_j \\
\pm z^*_j & 0 \\
\end{pmatrix}.
$ It's easy to see
$$
B_1(s, m+s; t, m+t) =
\begin{pmatrix}
0 & \pm z_j \\
\pm z_j^* & 0 \\
\end{pmatrix},
$$
which is the only Diagonal $2 \times 2$.
Thus, for column $i \in  [m] \cup \{m+s\} \setminus \{s, t\}$, $z_j$ in $B_1 (s, m+t)$ shares a Trivial $2 \times 2$, and for column $i \in \{s, m+1, \cdots, 2m\} \setminus \{m+s, m+t\}$, $z_j$ shares an Alamouti form, which implies
that all the $z[j]$'s in column $i \in A = [m] \cup \{m+s, m+t\} \setminus \{s, t\}$ are of the form $\pm z_j$, and in column $i \in \bar{A} = [2m] \setminus A$ are of the form $\pm z_j^*$.

After swapping column $s$ and $m+s$, column $t$ and $m+t$, we could move all the $\pm z_j$ (without conjugation) into the first $m$ columns,
which is $B_j$ form. Since $z_j$ is an arbitrary variable in $M_1$, by repeating this argument, we will exhaust all the $j \in [k]$, and the proof is complete.
\end{proof}

The following corollary is immediate from the above theorem.

\begin{corollary}
\label{cor:std_form}
 If $G$ is in standard form, then
\begin{itemize}
\item [(1)]
 For any $i = 1, 2, \ldots, m$, the $i$th column and $(m+i)$th column of $G$ have complement zero
patterns.
\item [(2)] For any $j \in [k]$, and any $i \in [m]$, the conjugations of $z[j]$ in column $i$ and column $m+i$
are different.
\item [(3)] For any $j \in [k]$, and any $i \in [m]$, $\pm z_j$ in column $i$, and $\pm z_j^*$ in column
$m+i$ (or $\pm z_j^*$ in column $i$, $\pm z_j$ in column $m+i$) form a Diagonal $2 \times 2$.
\end{itemize}
\end{corollary}
\begin{proof} It's clear that both (1), (2), (3) are true for each $B_j$ form. Since $G$ can be transformed into $B_j$
by column-restricted operations, and (1), (2), (3) are \emph{invariant} under column-restricted operations, we conclude that
$G$ satisfies (1), (2) and (3).
\end{proof}

\section{Minimal decoding delay}

The following lemma is the key result for proving the delay lower bound, which says that the existence of left zero
pattern $\alpha$ implies the existence of zero pattern $\beta$, where $\beta$ is obtained by changing two arbitrary bits of
$\alpha$ from $0$ to $1$.
\begin{lemma}
\label{lem:induce}
Let $G$ be a $[2k, 2m, k]$ BCOD in the standard form. Let $r$ be one row in $G$ with left zero pattern $\alpha \in \mathbb{F}_2^{m}$ and left weight $0 \le u \le m-2$. Then for any distinct $i, j \in [m]$ such that $\alpha(i) = \alpha(j) = 0$, there exists some
row in $G$ with left zero pattern $\alpha \oplus e_i \oplus e_j$, and the same conjugation, and thus has left weight $u+2$.
\end{lemma}
\begin{proof} Denote by $I_1$ the support of $\alpha$, that is,
$$
I_1 = \{ i \in [m] : \alpha(i) = 1 \},
$$
and $I_0 = [m] \setminus I_1$. Let $i, j \in I_0$ be distinct. We will show that there exists some row with
left zero pattern $\alpha \oplus e_i \oplus e_j$.

Suppose row $r$ has left zero pattern $\alpha$. By (1) of Corollary \ref{cor:std_form}, $r$ has right zero pattern $\overline{\alpha}$.
Without loss of generality, assume $r(m+i) = z_1$. By (3) of Corollary \ref{cor:std_form}, we claim that there exists
$\pm z_1^*$ in column $i$. Because $z_1$ in row $r$ is contained in $m-1$ Trivial $2 \times 2$, $m-1$ Alamouti $2 \times 2$,
and one Diagonal $2 \times 2$, for all the $l \in I_0\setminus \{i \}$, column $l$ contains variable $\pm z_1$ (not its conjugation), and for all the $l \in I_1$, column $l$ contains variable $\pm z_1^*$.

Take the row $r'$ such that $\pm z_1$ on the $j$th column, i.e., $r'(j) = \pm z_1$. We claim the left zero pattern of
$r'$ is exactly what we need. We need to verify that
\begin{itemize}
\item[(1)] $r'(i)$ is not zero.
\item[(2)] For all $l \in I_1$, $r'(l)$ it not zero.
\item[(3)] For all $l \in I_0 \setminus \{i, j\}$, $r'(l) = 0$.
\end{itemize}

For (1), notice that $r'(j) = \pm z_1$, and there exists some row $r_c$ such that $r_c(i) = \pm z_1^*$. Considering
the $2 \times 2$ submatrix formed by the $i$th and $j$th column of row $r'$ and $r_c$, they are either Alamouti $2\times 2$
or Diagonal $2 \times 2$. Again by (3)  in Corollary \ref{cor:std_form}, they must be Alamouti $2 \times 2$, which implies
$r'(i) \not= 0$.

For (2), take any $l \in I_1$. We know that there exist some row $r_l$ such that $r_l(l) = \pm z_1^*$. Considering
the $2 \times 2$ submatrix formed by the $l$th and $j$th column of row $r'$ and $r_l$, they are either Alamouti $2\times 2$
or Diagonal $2 \times 2$. By (3)  in Corollary \ref{cor:std_form}, they must be Alamouti $2 \times 2$, which implies
$r'(l) \not= 0$.

For (3), take any $l \in I_0 \setminus \{i, j\}$. There exist some row $r_l$ such that $r_l(l) = \pm z_1$. Considering
the $2 \times 2$ submatrix formed by the $l$th and $j$th column of row $r'$ and $r_l$, it must be Trivial $2\times 2$, which implies $r'(l) = 0$.
\end{proof}

Now, we are ready to prove the lower bound on the decoding delay of BCOD for $n = 2m$, $m$ odd. For the case
$m \equiv 1, 2, 3 \pmod 4$, the lower bound $2^m$ is proved in \cite{ADK11}. For the case $m \equiv 0 \pmod 4$, they
are able to prove lower bound $2^{m-1}$.

\begin{theorem}
\label{thm:delay_odd}
Let $G$ be a $[2k, 2m, k]$ BCOD with $m$ odd. Then $2k \ge 2^m$.
\end{theorem}
\begin{proof} It suffices to prove that every possible left zero pattern exists. Since a left zero pattern is a vector
in $\mathbf{F}_2^m$, there are $2^m$ in total.

Assume $G$ is already in its standard form $B_1$. From $B_1$ form, we claim that all left zero
patterns with weight $1$ exist. By Lemma \ref{lem:induce}, all left zero patterns with weight $3$ exist since each time
we can transform any two zero entries to ones. Repeating this argument, all left zero patterns of the same conjugation with weight $1, 3, 5, \ldots, m$ exist. Notice that if one row has left weight $u$, then its complement has weight $m-u$. From this observation, we conclude
all left zero patterns exists, which completes the proof.
\end{proof}

For the case $n = 2m$, $m$ even, the proof is similar, except that we need to take the conjugation into account. And
the following theorem proves the conjecture that, for $n = 2m$ congruent to $0$ module $8$, the delay $2k$ is
lower bounded by $2^m$.

\begin{theorem} Let $G$ be a $[2k, 2m, k]$ BCOD with $m$ even. Then $2k \ge 2^m$.
\end{theorem}
\begin{proof}
By the same argument as we did in Theorem \ref{thm:delay_odd}, we claim that all the left zero patterns with weight
$1, 3, \ldots, m-1$ exists, where the amount is
$$
{m \choose 1} + {m \choose 3} + \ldots + {m \choose m-1} = 2^{m-1}.
$$
By the definition of BCOD, $G$ is conjugation separated. Observe that any row has different conjugation with its
complement. Therefore, we claim that for every $\alpha \in \mathbb{F}^m$, there exists at least two rows $r, r'$ with
left zero pattern $\alpha$ and different conjugations, which implies that the number of rows $2k \ge 2^m$.
\end{proof}

We would like to point out that, in \cite{ADK11}, the lower bound for $\nu(n)$ is proved by reducing BCOD $[2k, 2m, k]$
to \textit{Real Orthogonal Design} (ROD) with parameter $[2k, 2m, 2k]$. It's known that for ROD $[2k, 2m, 2k]$,
the delay $2k$ is lower bounded by $\nu(n)$ \cite{Hur98, Hur23, Rad22}, where $\nu(n) = 2^{\delta(n)}$ and
$$
\delta(n) = \begin{cases}
4t, & \text{if }n = 8t+1 \\
4t+1, & \text{if }n = 8t+2 \\
4t+2, & \text{if }n= 8t+3, 8t+4 \\
4t+3, & \text{if }n = 8t+5, 8t+6, 8t+7, 8t+8.
\end{cases}
$$
Our proof here is self-contained and combinatorial.

\textbf{Xiaodong Liu} received his B.S. degree in School of Mathematical Science from Fudan University, 2012.
Now, he is pursuing his Ph.D. degree in Computer Science in Fudan University.

\vspace{0.2cm}
\textbf{Yuan Li} received his B.S. degree in Computer Science from Fudan University, 2011.
Now, he is pursuing his Ph.D. degree in Computer Science in the University of Chicago.

\vspace{0.2cm}
\textbf{Haibin Kan} received the Ph.D. degree from Fudan University, Shanghai, China, 1999. After receiving the
Ph.D. degree, he became a faculty of Fudan University. From June 2002 to January 2006,
he was with the Japan Advanced Institute of Science and Technology as an assistant professor.
He went back Fudan University in February 2006, where he is currently a full professor. His
research topics include coding theory, complexity of computing, and information security.

\end{document}